\documentclass{amsart}

\newcommand*{\mailto}[1]{\href{mailto:#1}{\nolinkurl{#1}}}
\usepackage{hyperref}

\usepackage{amsfonts,amsmath,amsxtra,amsthm,amssymb,latexsym}

\usepackage{color}

\theoremstyle{plain}
\newtheorem{theorem}{Theorem}[section]
\newtheorem{corollary}[theorem]{Corollary}
\newtheorem{lemma}[theorem]{Lemma}

\newtheorem{remark}[theorem]{Remark}
\newtheorem{example}[theorem]{Example}
\newtheorem{proposition}[theorem]{Proposition}

\newcommand{\nn}{\nonumber}
\newcommand{\be}{\begin{equation}}
\newcommand{\ee}{\end{equation}}

\numberwithin{equation}{section}

 \DeclareMathOperator{\dom}{dom}
\DeclareMathOperator{\ran}{ran} 
\DeclareMathOperator{\ess}{ess}

\DeclareMathOperator{\ac}{ac}
\DeclareMathOperator{\diag}{diag}
\DeclareMathOperator{\loc}{loc}\DeclareMathOperator{\comp}{comp}

\newcommand\R{{\mathbb{R}}}
\newcommand\C{{\mathbb{C}}}

\newcommand\N{{\mathbb{N}}}
\newcommand\Z{{\mathbb{Z}}}

\newcommand\gH{{\mathfrak{H}}}
\newcommand\gS{{\mathfrak{S}}}

\newcommand{\gG}{{\Gamma}}

\newcommand{\gt}{\mathfrak{t}}
\newcommand{\gb}{\mathfrak{b}}

\newcommand{\gd}{{d}}
\newcommand{\gA}{{\alpha}}
\newcommand{\gB}{{\beta}}

\newcommand\cH{{\mathcal{H}}}
\newcommand\cI{{\mathcal{I}}}

\newcommand\rH{{\mathbf{H}}}

\newcommand\rD{{\rm{d}}}

\begin{document}

\title[Hamiltonians with point interactions]{1--D Schr\"odinger operators\\ with local point interactions: a review}

\author[A.\ Kostenko]{Aleksey Kostenko}
\address{Fakult\"at f\"ur Mathematik\\
Universit\"at Wien\\
Nordbergstr. 15\\
1090 Wien, Austria}
\email{\mailto{Oleksiy.Kostenko@univie.ac.at; duzer80@gmail.com}}

\author[M.\ Malamud]{Mark Malamud}
\address{Institute of Applied Mathematics and Mechanics\\
NAS of Ukraine\\ R. Luxemburg str. 74\\
Donetsk 83114\\ Ukraine}
\email{\mailto{mmm@telenet.dn.ua}}

\dedicatory{Dedicated with great pleasure to Fritz Gesztesy on the occasion of his 60th birthday.}
\thanks{The research was funded by the Austrian Science Fund (FWF) under project No.\ M1309--N13}
\thanks{{\it "Spectral Analysis, Integrable Systems, and Ordinary Differential Equations"}, H. Holden et
al. (eds), Proceedings of Symposia in Pure Mathematics {\bf 87}, Amer. Math. Soc. (to appear). }


\begin{abstract}
We review recent developments  in the theory of 1-D Schr\"odinger operators with local point interactions on a discrete set. The progress in this area was stimulated by recent advances in the extension theory of symmetric operators and in the theory of ordinary differential operators with distributional coefficients. 
\end{abstract}

\maketitle

\tableofcontents

\section{Introduction}\label{intro}

Schr\"odinger operators with potentials supported on a finite or a discrete set of points are known as solvable models of quantum mechanics. These models called "solvable" since their resolvents can be computed explicitly in terms of the interaction strengths and the location of the sources. As a consequence the spectrum, the eigenfunctions, and further spectral properties can be determined explicitly. Models of this type have been extensively discussed in the physical literature, mainly in atomic, nuclear and solid state physics. A comprehensive treatment of Schr\"odinger operators with point interactions as well as a detailed list of references can be found in the monograph \cite{AGHH88} published in 1988. In its second edition \cite{AGHH05}, published in 2005 by the American Mathematical Society, an account on the progress in the field for the period from 1988 until 2005 was summarized by Pavel Exner (see Appendix K "Seize ans apr$\grave{\rm{e}}$s" in \cite{AGHH05}). Our main aim is to review recent advances in the spectral theory of 1--D Schr\"odinger operators with local point interactions on a discrete set of points.

Historically, the first influential paper on 1--D Schr\"odinger operators with point interactions was the paper by Kronig and Penney \cite{KroPen31}. They considered the Hamiltonian 
\be\label{eq:KP}
\rH_{KP}=-\frac{d^2}{dx^2}+\sum_{k\in\Z}a\delta(x-k),
\ee
where $a\in\R$ is fixed and $\delta$ is  a Dirac delta-function.
This Hamiltonian, now known as "the Kronig--Penney model", describes a nonrelativistic  electron moving in a fixed crystal lattice.
Our main objects of interest are the following operators $\rH_{X,\gA,q}$ and $\rH_{X,\gB,q}$ associated with the formal differential expressions
\begin{eqnarray}\label{I_01A}
\ell_{X,\gA,q}:=-\frac{\rD^2}{\rD x^2}+q(x)+\sum_{x_{n}\in X}\gA_n\delta_n,\\
 \ell_{X,\gB,q}:=-\frac{\rD^2}{\rD x^2}+q(x)+\sum_{x_{n}\in X}\gB_n(\cdot,\delta'_n)\delta'_n, \label{I_01B}
\end{eqnarray}
where $\delta_n:=\delta(x-x_n)$.
These operators describe $\delta$- and $\delta'$-interactions, respectively, on a discrete set $X=\{x_n\}_{n\in I}\subset\cI=(a,b)$, and the coefficients $\gA_n,\ \gB_n\in\R$ are called the strengths of the interaction at the point $x=x_n$. Clearly, \eqref{eq:KP} is a particular case of \eqref{I_01A}  with $\cI=\R$, $X=\Z$, $\gA_n\equiv \gA$, and $q\equiv 0$.

The spectral properties of Hamiltonians associated with \eqref{I_01A} and \eqref{I_01B} are widely studied under the assumption that there is a positive uniform lower bound on the distance between interaction centers,
\be\label{eq:d>0}
\gd_*:=\inf_{i,j}|x_i-x_j|>0.
\ee
A comprehensive account on related results can be found in the monograph \cite{AGHH05}.
One of the main reasons for the assumption \eqref{eq:d>0} is that without this assumption even in the case $q\equiv 0$ the Hamiltonian \eqref{I_01A} might be non-self-adjoint, i.e., symmetric with nontrivial deficiency indices  (the first example was given by Shubin and Stolz in \cite{ShuSto94}). In the last few years this difficulty has been overcome due to recent advances in the extension theory of symmetric operators and in the theory of ordinary differential operators with distributional coefficients.

%

We would like to complete  the introduction with a few words about Fritz's work on point interactions.
It is difficult to overestimate his influence on the theory of Schr\"odinger operators with point interactions.
He is one of the founders and promoters of the spectral theory of Schr\"odinger operators with point interactions.
Under his influence, and with his participation over a long period, the subject has taken its present form.

Roughly speaking, his research in the field can be divided into two parts: (i) the study of Schr\"odinger operators with infinitely many interaction centers, and (ii) a rigorous definition of various classes of point interactions. His papers \cite{GesHol87} and \cite{GesSeb87} written jointly with Holden and \v Seba, respectively, originated a new concept of relativistic
and non-relativistic Hamiltonians with $\delta'$-interactions and had a long-year discussion in both physical and mathematical literature.
 It is also difficult to overestimate the role of the monograph \cite{AGHH88}, which
represents the foundation of a new and rapidly developing branch in the spectral theory of Schr\"odinger operators.

 Happy Birthday, Fritz, and many new  important and influential results!

{\bf Notation.}
$\N, \C, \R$ have the usual meaning; $\R_+=[0,\infty)$.

For a potential $q$ and sequences $\gA$ and $\gB$ we set $q^\pm(x):=(|q(x)|\pm q(x))/2$, $\gA_k^\pm:=(|\gA_k|\pm \gA_k)/2$, and $\gB_k^\pm:=(|\gB_k|\pm \gB_k)/2$.

For a self--adjoint
operator $T = T^*$ acting in a Hilbert space $\gH$,  $E_T(\cdot)$ denotes the spectral measure, $T^-:= TE_T(-\infty, 0)$ and $T^+:= TE_T(0,+\infty)$ are the negative and positive parts of $T$, respectively, and $\kappa_\pm(T):=\dim\big(\ran(T_\pm)\big)$ (if $\kappa_\pm(T)<\infty$, then $\kappa_\pm(T)$ is  the number of negative/positive eigenvalues of $T$ counting multiplicities). Further, $\sigma(T)$ and $\rho(T)$ are the spectrum and the resolvent set of $T$, respectively. By $\sigma_{\rm p}(T)$, $\sigma_{\rm pp}(T)$, $\sigma_{\ess}(T)$, $\sigma_{\ac}(T)$, and   $\sigma_{\rm{sc}}(T)$ we denote point, pure point, essential, absolutely continuous
and singular continuous spectra of $T$.

Let $X$ be a discrete subset of $\R_+$, $X=\{x_k\}_1^\infty$ and $x_k\uparrow +\infty$.
Also we shall use the following Sobolev spaces ($n\in\N$ and $p\in[1,\infty]$)
%
%
\begin{eqnarray*}
&W^{n,p}(\R_+\setminus X):=\{f\in L^p(\R_+): f\in W^{n,p}[x_{k-1},x_k],\  k\in \N,\,  f^{(n)}\in
  L^p(\R_+)\},\\
&W^{n,p}_0(\R_+\setminus X):=\{f\in
W^{n,p}(\R_+): f(x_k)=...=f^{(n-1)}(x_k)=0,\, k\in \N\},\\
&W^{n,p}_{\comp}(\R_+\setminus X):=W^{n,p}(\R_+\setminus X)\cap L^p_{\comp}(\R_+).
\end{eqnarray*}

\section{Hamiltonians with $\delta$-interactions}\label{Sec:II}

\subsection{Definition of $\delta$-interactions}\label{sec:2.1} There are several ways to associate an operator with the differential expression $\ell_{X,\gA,q}$. For example, a $\delta$-interaction at a point $x=x_0$ may be defined using the \emph{form method}, that is the operator $-\frac{\rD^2}{\rD x^2}+\gA_0\delta(x-x_0)$
is defined as an operator associated in $L^2(\R)$ with the quadratic form 
\[
\mathfrak{t}[f]:=\int_{\R}|f'(t)|^2dt+\gA_0|f(x_0)|^2,\qquad \dom(\gt):= W_2^1(\R),
\]
which is closed and lower semibounded by the KLMN Theorem (see \cite[p. 168]{ReeSim75}).
Another way to introduce a local interaction at $x_0$ is to consider a symmetric operator $\rH_{\min}:=\rH_{\min}^-\oplus \rH_{\min}^+$, where $\rH_{\min}^-$ and $\rH_{\min}^+$ are the minimal operators generated by $-\frac{\rD^2}{\rD x^2}$ in $L^2(-\infty,x_0)$ and $L^2(x_0,+\infty)$, respectively, and then to consider its extension subject to \emph{the boundary conditions} connecting $x_0+$ and $x_0-$:\be
f(x_0+)=f(x_0-),\quad f'(x_0+)-f'(x_0-)=\gA_0f(x_0).
\ee

Both these methods have disadvantages if the set $X$ is infinite. The form method works only for the case of lower semibounded operators. If we apply the method of boundary conditions, then the corresponding minimal operator $\rH_{\min}$ has infinite deficiency indices and the description of self-adjoint extensions of $\rH_{\min}$ is a rather complicated problem in this case.

In what follows, without loss of generality we shall consider $\ell_{X,\gA,q}$ on the positive semi-axis $\R_+$ assuming that the sequence $X=\{x_k\}_{k=1}^\infty$ is strictly increasing, $0=x_0<x_1<x_2<..<x_k<x_{k+1}<...$, and satisfies $x_k\uparrow +\infty$. We shall define the Hamiltonian with $\delta$-interactions on $X$ as follows: let
\begin{align}
\rH_{X,\gA,q}^0f:=\tau_qf=&-f''+q(x)f,\quad f\in\dom(\rH_{X,\gA,q}^0),\label{eq:h_a1}\\
\dom(\rH_{X,\gA,q}^0):=\Big\{f\in & W^{2,1}_{\comp}(\R_+\setminus X):\,  f(0)=0, 
\Big. \nonumber
\\ &\Big.\begin{array}{c}
f(x_k+)=f(x_k-)\\ f'(x_k+)-f'(x_k-)=\gA_k f(x_k) \end{array},\, \tau_q f\in L^2(\R_+)\Big\}.\label{eq:h_a2}
\end{align}
Clearly, the operator $\rH_{X,\gA,q}^0$ is symmetric. Let us denote its closure  by $\rH_{X,\gA,q}$:
\be\label{eq:h_a3}
\rH_{X,\gA,q}:=\overline{\rH_{X,\gA,q}^0}.
\ee
 If $q=\bold{0}$, we set $\rH_{X,\gA}:=\rH_{X,\gA,0}$; if either $X=\emptyset$ or $\gA=\bold{0}$, then $\rH_{X,\gA,q}$ will be denoted by $\rH_{q}$.

An alternative approach was proposed by A. Savchuk and A. Shkalikov in \cite{SavShk99} (see also \cite{SavShk03})\footnote{This regularization method was used in \cite{AtkEveZet88} in the particular case
$q(x) = 1/x$ and then further developed for generic $W^{-1,2}_{\loc}$-distributional potentials in \cite{SavShk99}, \cite{SavShk03}.
For further historical remarks we refer to \cite{EckGesNicTes12}, \cite{EckGesNicTes12b}, \cite{GorMih10}, and \cite{GorMih10b}.}.
Namely, they suggested to consider \eqref{I_01A} with the help of quasi-derivatives.
The potential $v(x)=q(x)+\sum_{k=1}^\infty \gA_k\delta(x-x_k)$ is
a derivative of the function $V(x)=\int_0^x q(t)dt+ \sum_{x_k<x}\gA_k$ in the sense of  distributions. Therefore, we can rewrite the differential expression \eqref{I_01A} as follows
\be\label{eq:2.2}
\ell_{X,\gA,q}\, y=\ell_{V'}\, y:=-(y^{[1]})'-V(x)y^{[1]}-V^2(x)y,\quad y^{[1]}:=y'-V(x)y,
\ee
and then define the operator $\rH_{V'}^0$ by \eqref{eq:2.2} on the domain  
\be\label{eq:2.3}
\dom(\rH_{V'}^0)=\{f\in L^2_{\comp}(\R_+):\, f,f^{[1]}\in AC_{\loc}(\R_+),\, f(0)=0,\, \ell_{V'}f\in L^2(\R_+)\}.
\ee
It is straightforward to check that the  operators $\rH_{X,\gA,q}^0$ and $\rH_{V'}^0$ defined by \eqref{eq:h_a1}--\eqref{eq:h_a2} and by \eqref{eq:2.2}--\eqref{eq:2.3}, respectively, coincide (see \cite{SavShk03}).
This definition preserves the main features of the classical Sturm--Liouville theory (for instance, Weyl--Titchmarsh theory \cite{EckTes12}). Moreover, it allows ones to study the direct and inverse spectral problems for 1--D Schr\"odinger operators with potentials distributions. We decided not to discuss this topic here in order to keep our review at a reasonable length (however, see remarks at the end of this section). 

Let us emphasize that definition \eqref{eq:2.2}--\eqref{eq:2.3} is applicable only under the assumption that $V\in L^2_{\loc}(\R_+)$, that is, the original potential is a $W^{-1,2}_{\loc}$-distribution (see discussion in \cite{SavShk03}). Clearly, this excludes the case of $\delta'$-interactions.

Next assuming that $q\in L^\infty(\R_+)$, we specify a description of the domains $\dom(\rH_{X,\gA,q})$ and $\dom(\rH_{X,\gA,q}^*)$
equipped with the graph norms of the operators $\rH_{X,\gA,q}$ and $\rH_{X,\gA,q}^*$, respectively.
   \begin{proposition}\label{prop2.1}
Let $q\in L^\infty(\R_+)$. Then:

\begin{itemize}
\item[(i)]  The operator $\rH_{X,\gA,q}$  is symmetric and its adjoint 
 is given by the same differential expression $\tau_q$ on the domain
 \be
 \dom\big(\rH_{X,\gA,q}^*\big)=\Big\{f\in  W^{2,2}(\R_+\setminus X): f(0)=0, 
\begin{array}{c}
f(x_k+)=f(x_k-)\\ f'(x_k+)-f'(x_k-)=\gA_k f(x_k) \end{array}\Big\}.\label{eq:2.5B}
 \ee
\item[(ii)] Assume additionally that  $\gd_* >0$,  
\be
\gd_*:=\inf_{k}\gd_k,\quad  \text{and} \quad \gd_k:=x_k-x_{k-1},\quad k\in\N.
\ee 
 Then
the embedding
  \begin{align}\label{2.8}
\dom(\rH_{X,\gA,q}) \hookrightarrow \dom(\rH_{X,\gA,q}^*)  \hookrightarrow  W^{1,2}(\R_+)
    \end{align}
holds  and is continuous.
\end{itemize}
   \end{proposition}
\begin{proof}
(i) follows from integration by parts of the expression $(\rH_{X,\gA,q}f, g)$.

(ii) If $\gd_* >0$, then applying the Sobolev embedding theorem to the spaces $W^{2,2}[x_{k-1},x_k]$,  $k\in\N$ (see \cite[inequality (IV.1.12)]{Kato66}  and also the proof of \cite[Proposition 2.1(ii)]{KosMal12}), we conclude that $W^{2,2}(\R_+\setminus X) = \bigoplus^{\infty}_{k=1}W^{2,2}[x_{k-1},x_k]$ is continuously embedded into
$W^{1,2}(\R_+\setminus X) = \bigoplus^{\infty}_{k=1}W^{1,2}[x_{k-1},x_k]$. 
The description \eqref{eq:2.5B} of $\dom(\rH_{X,\gA,q}^*)$ completes the proof.
   \end{proof}
\begin{remark}\label{rem:2.1}
Let us stress that in  the case $d_*=0$ the embedding \eqref{2.8} depends on  $\gA$ and might be false (see Example \ref{ex:2.2} and also Remark \ref{rem:2.16}(iii) below).
\end{remark}

\begin{example}\label{ex:2.2}
Let $X=\{x_k\}_{k=1}^\infty$ be such that $\gd_{2k-1}=\gd_{2k}=\frac{1}{k}$, $k\in\N$. Let also  $\gA_k= \frac{2}{\gd_k}$, $k\in\N$. Consider the Hamiltonian
\be
\rH:= \rH_{X,\gA,q}= -\frac{d^2}{dx^2}-\sum_{k=1}^\infty \frac{2}{\gd_k}\delta(x-x_k)
\ee
Define the function $f:\R_+\to \R$ as follows: $f(x)=x$ on $[0,1]$; $f(x)=x_{4k-2}-x$ if $x\in [x_{4k-3},x_{4k-1}]$ and $f(x)=x-x_{4k}$ if $x\in[x_{4k-1},x_{4k+1}]$, $k\in\N$. Clearly, $f''(x)=0$ for all $x\in \R_+\setminus X$ and
\[
\int_0^\infty |f(x)|^2dx=\sum_{k=1}^\infty\frac{\gd_k^3}{3} =\sum_{k=1}^\infty\frac{2}{3k^3}<\infty.
\]
Moreover, it is straightforward to check that the function $f$ satisfies boundary conditions \eqref{eq:h_a2} for all $k\in\N$. Therefore,
$f\in \dom(\rH^*)$.  However, $f'(x)=1$ for all $x\in\R_+\setminus X$ and hence $f\notin W^{1,2}(\R_+\setminus X)$.

Note that the operator  $\rH$ is not lower semibounded (see Theorem \ref{th:brinck}(ii)).
Moreover, it can be checked that it is symmetric with $n_{\pm}({\rH}) =1.$

Apparently, it is possible to construct examples of self-adjoint Hamiltonians ${\rH}_{X,\gA} = {\rH}^*_{X,\gA}$ such that the embedding \eqref{2.8} fails.
  \end{example}
   {\em Further references:} an extension of the Savchuk--Shkalikov approach to the case of more general Sturm--Liouville equations, as well as to operators with matrix-valued coefficients, can be found in  
 \cite{EckTes12}, \cite{EckGesNicTes12}, \cite{EckGesNicTes12b}, \cite{GorMih10}, \cite{GorMih10b}, \cite{MirSaf11}.
 
 Sturm--Liouville operators on finite intervals with singular potentials have also been considered in the framework of the inverse spectral theory. In particular, the inverse spectral problems of reconstruction of the potential from the corresponding spectral data (from two spectra or one spectrum and the set of norming constants) have successfully been solved in the paper \cite{Zhi67}
for potentials that are signed measures and in \cite{HryMyk03} and \cite{SavShk05} for potentials that are distributions in $W_2^{-1}$. Sturm--Liouville operators in impedance form, i.e., of the form
\[
-a^{-2}(x)\frac{d}{dx}a^2(x)\frac{d}{dx}
\]
with a positive impedance function $a$, were discussed in the papers \cite{And88a}, \cite{And88b}, \cite{RunSac92}, \cite{ColMcL93}, \cite{AlbHryMyk05}. For a regular enough function $a$, such an operator (under, say, the Dirichlet boundary conditions) is unitarily equivalent to a Sturm--Liouville operator in a potential form with the potential $q = a''/a$. The inverse spectral problem for impedance Sturm--Liouville  operators with $a$ of bounded variation was studied in \cite{And88a}; note that then the corresponding $q$ formally contains  singularities of the form $\delta'$. The case $a\in W_2^1$ was also completely analyzed in \cite{And88b}, \cite{RunSac92}, \cite{ColMcL93}, while $a\in W_p^1$ with $p\ge1$ in \cite{And88a} (partially) and in \cite{AlbHryMyk05}. In the papers \cite{SavShk10} and \cite{Hry11} the global uniform stability in the inverse spectral problem of reconstruction of singular Sturm--Liouville operators from either two  spectra or one spectrum and the norming constants is established; the potentials are from the Sobolev spaces $W_2^{\theta}$ with $\theta\ge-1$.

\subsection{Self-adjointness} \label{sec:2.2}

In the seminal paper \cite{GesKir85}, Gesztesy jointly with Kirsch proved the following very important result.

\begin{theorem}[\cite{GesKir85}]\label{th:gesztesy}
Let the Hamiltonian $\rH_{X,\gA,q}$ be defined by \eqref{eq:h_a1}--\eqref{eq:h_a2}. Assume that the set $X$ satisfies \eqref{eq:d>0} and the potential $q\in L^1_{\loc}(\R)$ is such that for any $\varepsilon<\gd_*/2$ the negative part of the potential
\be
q_\varepsilon(x):=q(x)\chi_{\varepsilon}(x),\quad \chi_\varepsilon(x):=\begin{cases}
1, & x\in \cup_{k=0}^\infty(x_k+\varepsilon,x_{k+1}-\varepsilon)\\
0, & x\notin \cup_{k=0}^\infty(x_k+\varepsilon,x_{k+1}-\varepsilon)
\end{cases},
\ee
is form-bounded with respect to the free Hamiltonian $H_0=-\frac{d^2}{dx^2}$ with relative bound $a_\varepsilon<1$. Then $\rH_{X,\gA,q}$ is self-adjoint.
\end{theorem}

\begin{corollary}[\cite{GesKir85}]\label{cor:2.2}
If $q$ is lower semibounded, $q(x)\ge -c$ a.e. on $\R_+$, and \eqref{eq:d>0} holds true, then the operator $\rH_{X,\gA,q}$ is self-adjoint.
\end{corollary}

\begin{remark}
If $X$ is unbounded, previous investigations of  Hamiltonians with $\delta$-interactions
either used the resolvent of $\rH_{X,\gA,q}$ (see \cite[\S III.2]{AGHH05} and references therein) or the technique of
local partitions \cite{Mor79} in order to define $\rH_{X,\gA,q}$ by the method of forms. In the
one-dimensional case, Theorem \ref{th:gesztesy} provides a powerful alternative to such methods
which even applies if Hamiltonians $\rH_{X,\gA,q}$ unbounded from below are involved.
\end{remark}

It turned out that both assumptions on the negative part of the potential and on the set $X$ are essential. If the potential $q$ is "very negative", then one needs to make an additional assumption on interaction strengths $\gA$ in order to ensure the self-adjointness of the Hamiltonian $\rH_{X,\gA,q}$. 

\begin{theorem}[\cite{ShuSto94}]\label{th:stolz}
Let the Hamiltonian $\rH_{X,\gA,q}$ be defined by \eqref{eq:h_a1}--\eqref{eq:h_a2}. Assume that the set $X$ satisfies \eqref{eq:d>0} and there are positive constants $C_1$, $C_2$, $C_3$, $C_4$ such that
\be
q(x)\ge -C_1x^2-C_2,\quad \gA_k\ge -C_3|x_k|-C_4.
\ee
Then $\rH_{X,\gA,q}$ is self-adjoint.
\end{theorem}

If the set $X$ does not satisfy \eqref{eq:d>0}, then, as it was first observed  by Shubin and Stolz \cite{ShuSto94}, the Hamiltonian $\rH_{X,\gA,q}$ might be symmetric with the nontrivial deficiency indices even in the case of zero potential $q$. Namely (see \cite[p. 496]{ShuSto94}), they proved that the Hamiltonian
\be\label{eq:2.9}
\rH=-\frac{d^2}{dx^2}-\sum_{k=1}^\infty (2k+1)\delta(x-x_k), \quad x_{k+1}-x_k=\frac{1}{k},
\ee
is symmetric with deficiency indices $n_\pm(\rH)=1$.

\subsection{Connection with Jacobi matrices}\label{ss:2.3}

The progress on the case $\gd_*=0$ was made by the authors in the recent papers \cite{KosMal09}, \cite{KosMal09b}, \cite{KosMal10} in the framework of extension theory of symmetric operators. The main tool in \cite{KosMal09} is the concept of boundary triplets and the corresponding Weyl functions (see \cite{Gor84}, \cite{DM91}, \cite{BruGeyPan07}). The main ingredient of this approach
is the following abstract version of the  Green formula for the adjoint $A^*$ of a symmetric operator $A\in\mathcal{C}(\gH)$,
  \begin{equation}\label{II.1.2_green_fIntro}
(A^*f,g)_\gH - (f,A^*g)_\gH = (\gG_1f,\gG_0g)_\cH -
(\gG_0f,\gG_1g)_\cH, \qquad f,g\in\dom(A^*).
\end{equation}
Here $\cH$ is an auxiliary Hilbert space and  the mapping $\Gamma :=\begin{pmatrix}\Gamma_0\\\Gamma_1\end{pmatrix}:  \dom(A^*)
\rightarrow \cH \oplus \cH$ is required to be surjective.
The mapping $\Gamma$ leads to  a natural  parametrization   of self-adjoint (symmetric)
extensions of $A$ by means of self-adjoint (symmetric) linear relations in $\cH$, see \cite{Gor84, DM91}.
For instance, every extension  $\widetilde A = \widetilde A^*$, which is  disjoint with $ A_0:=A^*\!\upharpoonright\ker(\Gamma_0)$,
admits a representation
   \begin{equation}\label{eq:A_B}
\widetilde A = A_B:=A^*\!\upharpoonright \ker\bigl(\Gamma_1 - B\Gamma_0\bigr)\quad
     \end{equation}
where $\quad B=B^*  \in \mathcal{C}(\mathcal{H})$ is the "boundary" operator and its graph in $\cH$ is given by
 $\Gamma\dom(\widetilde A) := \{\{\Gamma_0f, \Gamma_1f\}: f \in \dom(\widetilde A)\}$.
As distinguished from the J. von Neumann approach, \eqref{eq:A_B}  yields
a natural parametrization  of all self-adjoint (symmetric) extensions
directly in terms of (abstract) boundary condition.

Assuming $q\in L^\infty(\R_+)$, we  consider the Hamiltonian  $\rH_{X,\gA,q}$ as an extension of the minimal symmetric operator
$\rH_{X,q} = \bigoplus_{k\in \N} \rH_{q,k}$, where
\[
\rH_{q,k}f:=\tau_qf = -f'' +qf,\quad \dom(\rH_{q,k})=  W^{2,2}_0[x_{k-1},x_k].
\]
To construct an appropriate  boundary triplet $\Pi=\{\cH,\Gamma_0,\Gamma_1\}$ for the operator $\rH_{X,q}^* = \bigoplus_{k\in \N} \rH_{q,k}^*$ we apply the construction
elaborated in \cite{MalNei_08} and \cite{KosMal09} (note that  a direct sum of boundary triplets is not necessarily a boundary triplet if either  $\gd_*=0$ or $q\notin L^\infty(\R_+)$, see \cite{Koc79} and \cite{KosMal09}). Based on this construction, it is shown that the domain of $\rH_{X,\gA,q}$ admits the following representation 
\begin{align}
&\qquad\, \qquad \dom(\rH_{X,\gA,q})=\ker(\Gamma_1-B_{X,\gA}\Gamma_0),\nn\\
B_{X,\gA}=&\left(\begin{array}{cccc}
r_1^{-2}\bigl(\alpha_1+\frac{1}{\gd_1}+\frac{1}{\gd_2}\bigr) & (r_1r_2\gd_2)^{-1} & 0&   \dots\\
(r_1r_2\gd_2)^{-1} &r_2^{-2}\bigl(\alpha_2+\frac{1}{\gd_2}+\frac{1}{\gd_3}\bigr) & (r_2r_3\gd_3)^{-1} &  \dots\\
0 & (r_2r_3\gd_3)^{-1} & r_3^{-2}\bigl(\alpha_3+\frac{1}{\gd_3}+\frac{1}{\gd_4}\bigr)&  \dots\\
\dots & \dots & \dots& \dots
\end{array}\right), \label{eq:B_xa}
\end{align}
  and $r_n=\sqrt{\gd_n+\gd_{n+1}}$, $n\in\N$.
This parameterization implies that certain spectral properties of the operator  $\rH_{X,\gA,q}$ correlate with the corresponding spectral properties of the Jacobi matrix $B_{X,\gA}$. Namely,  the following result was established in \cite{KosMal09}, \cite{KosMal09b}.

\begin{theorem}[\cite{KosMal09, KosMal09b}]\label{th:km_sa}
Let $q\in L^\infty(\R_+)$ and let $X= \{x_k\}_{k=1}^\infty\subset \R_+$  be strictly increasing and such that $\gd^*:=\sup_k \gd_k<\infty$. Then:
\begin{itemize}
\item[(i)] The deficiency indices of the operators $\rH_{X,\gA,q}$ and $B_{X,\gA}$ coincide and
   \be
   n_\pm(\rH_{X,\gA,q})=n_\pm(B_{X,\gA})\le 1.
   \ee
 In particular,   the operator $\rH_{X,\gA,q}$ is self-adjoint if and only if so is the Jacobi matrix $B_{X,\gA}$. 
\item[(ii)] The operator $\rH_{X,\gA,q}$ is lower semibounded precisely if so is the matrix $B_{X,\gA}$.
\item[(iii)] The self-adjoint operator $\rH_{X,\gA,q}$ has purely discrete spectrum if and only if $\gd_k\to 0$ as $k\to \infty$ and the spectrum of $B_{X,\gA}$ is purely discrete.
\item[(iv)] If the operator $\rH_{X,\gA}$ is self-adjoint, then
\be
\kappa_-(\rH_{X,\gA})=\kappa_-(B_{X,\gA}).
\ee
In particular, the operator $\rH_{X,\gA}$ is nonnegative if and only if so is the matrix $B_{X,\gA}$.
\item[(v)] If the operator $\rH_{X,\gA}$ is self-adjoint, then for any $p\in [1,\infty]$
\be
\rH_{X,\gA}^-\in \gS_p(L^2) \Longleftrightarrow B^-_{X,\gA}\in \gS_p(l^2).
\ee
In particular, $\sigma_{\ess}(\rH_{X,\gA,q})\subseteq\R_+$ if and only if $\sigma_{\ess}(B_{X,\gA})\subseteq\R_+$.
\end{itemize}
\end{theorem}

\begin{remark}
(i) In the case $\gd_*>0$, the boundary triplets approach was first applied for the study of spectral properties of Hamiltonians with local point interactions by Kochubei \cite{Koc89} (see also Mikhailets \cite{Mih94}).

(ii) Let us mention that using a different approach the inequality $n_\pm(\rH_{X,\gA,q})\le 1$ was first established by Minami \cite{Min87} for arbitrary potentials $q$, not necessarily bounded (see also \cite{BusStoWei95} and \cite{ShuSto94}).

(iii) In the case $\gd_*>0$, the connection between Hamiltonians with point interactions and Jacobi matrices goes back at least to the papers by Phariseau \cite{Pha60}, \cite{Pha60b} and Bellissard et. al. \cite{Bel82} (for further details we refer to \cite[\S III.2]{AGHH05}).
\end{remark}

Theorem \ref{th:km_sa} allows us to apply the well developed spectral theory of Jacobi operators (see, e.g.,  \cite{Akh}, \cite{Ber68}, \cite{KosMir99}, \cite{KosMir01}, \cite{Tes99}) for the study of spectral properties of operators $\rH_{X,\gA,q}$. For instance, applying the Carleman test (see \cite{Akh}) to the matrix $B_{X,\gA}$, we immediately obtain the following improvement of Corollary \ref{cor:2.2} in the case $q\in L^\infty(\R_+)$.

\begin{corollary}\label{cor:2.6}
Let $q\in L^\infty(\R_+)$ and let $X$ be such that $\sum_{k=1}^\infty \gd_k^2=\infty$. Then $\rH_{X,\gA,q}$ is self-adjoint.
\end{corollary}

Let us mention that the condition $\{\gd_k\}\notin l^2$ is sharp (see \cite[Proposition 5.9]{KosMal09}). Let us conclude this subsection with the following example (see \cite[Example 5.12]{KosMal09} and also \cite[Proposition 3]{KosMal09b}).

\begin{example}\label{ex:2.8}
Let $\mathcal{I}=\R_+,\ x_0=0, \ x_{k}-x_{k-1}=\gd_k:=1/k,\ k\in\N $. Consider the operator
\begin{equation}\label{IV.2.2_11}
\rH_{A}:=-\frac{\rD^2}{\rD x^2}+\sum_{k=1}^\infty\alpha_k\delta(x-x_k).
      \end{equation}
Clearly,  $\{\gd_k\}_{k=1}^\infty\in l_2$ and we can not apply Corollary \ref{cor:2.6}. However, the following statements are true:
\begin{enumerate}
\item[(i)]  \ If $\sum_{k=1}^\infty
\frac{|\alpha_k|}{k^3}=\infty$, then the  operator $\rH_{A}$  is
self-adjoint.
\item[(ii)]  \ If $\alpha_k\leq -2(2k+1)+O(k^{-1})$,
then $\rH_{A}$ is self-adjoint.
\item[(iii)]  \ If $\alpha_k\geq
-\frac{C}{k},\ k\in \N,\ C\equiv const>0$, then $\rH_{A}$ is
self-adjoint.
\item[(iv)]  \ If $\alpha_k=-(2k+1)+O(k^{-\varepsilon})$ with some $ \varepsilon>0$, then $\mathrm{n}_\pm(\rH_A)=1$.
\item[(v)] \ If $\gA_k=-a(2k+1)+O(k^{-1})$ with some $a\in (0,2)$, then $\mathrm{n}_\pm(\rH_A)=1$.
    \end{enumerate}
\end{example}

\begin{remark}
Example \ref{ex:2.8} is inspired by the example of Shubin and Stolz (cf. \eqref{eq:2.9}) and its proof is
based on various self-adjointess tests for Jacobi matrices (cf. \cite{Akh} and \cite{Ber68}).
In particular, the proof of (iv) is based on the recent improvement by Kostyuchenko and Mirzoev \cite[Theorem 1]{KosMir01} of the well-known
Berezanskii condition \cite[Theorem VII.1.5]{Ber68}.
\end{remark}

{\em Further notes:} A generalization and further developments of Example \ref{ex:2.8} can be found in \cite{KarTys12}. Using the approach based on quasi-derivatives, it was noticed in  \cite{IsmKos10} and \cite{MirSaf11}, \cite{Kon11} that the analysis of \cite{Ism63}, \cite{Ism85} and \cite{Mir91} extends to the case of Hamiltonians with $\delta$-interactions. In particular, using this approach one can extend  Corollary \ref{cor:2.6} to the case of semibounded potentials $q$, $q(.)\ge -c$ a.e. on $\R_+$.

\subsection{Semiboundedness}\label{sec:2.3} As we already mentioned in Section \ref{sec:2.1}, the Hamiltonian $\rH_{X,\gA,q}$ may  be defined via the energy form
\begin{align}
\gt_{X,\gA,q}^0[f]:=\int_{\R_+}(|f'(x)|^2+q(x)|f(x)|^2)\, dx +\sum_{k=1}^\infty \gA_k|f(x_k)|^2,\label{eq:2.12}\\
\quad \dom(\gt^0_{X,\gA,q})=\{f\in W^{1,2}(\R_+)\cap L^2_{\comp}(\R_+):\, \gt_{X,\gA,q}^0[f]<\infty\}.\label{eq:2.13}
\end{align}
Clearly, this form admits the representation
\be\label{eq:2.14}
\gt_{X,\gA,q}^0[f]=(\rH_{X,\gA,q}^0f,f)_{L^2},\quad f\in 
\dom(\rH_{X,\gA,q}^0).
\ee
So, one is interested in conditions on $X$ and $\gA$ such that the form $\gt_{X,\gA,q}^0$ is lower semibounded (and hence closable) and then to describe its closure.

\begin{theorem}[\cite{AlbKosMal09}]\label{th:povzner}
If the Hamiltonian $\rH_{X,\gA,q}$ is lower semibounded, then it is self-adjoint. In particular, if the form $\gt_{X,\gA,q}^0$ is lower semibounded, then it is closable and the self-adjoint operator associated with its closure $\gt_{X,\gA,q}:=\overline{\gt_{X,\gA,q}^0}$ coincides with $\rH_{X,\gA,q}$.
\end{theorem}

\begin{remark}
Theorem \ref{th:povzner} is the analog of the celebrated Glazman--Povzner--Wienholtz Theorem \cite{Ber68}, \cite{Gla65}, \cite{Wie58} (see also the paper \cite{ClaGes03} by Clark and Gesztesy, where the case of matrix-valued Schr\"odinger operators was treated). An alternative proof of Theorem \ref{th:povzner} has  recently been proposed by Hryniv and Mykytyuk \cite{HryMyk12}. Let us also mention that a connection between lower-semiboundedness and self-adjointness for general Sturm--Liouville operators was first observed by Hartman \cite{Har48} and Rellich \cite{Rel51}. 
Further details as well as a comprehensive list of references can be found in \cite{ClaGes03}.
\end{remark}

The following result was obtained by Brasche in \cite{Bra85}.

\begin{theorem}[\cite{Bra85}]\label{th:brinck}
Assume that the negative parts of the potential $q$ and the sequence $\gA$ satisfy the following conditions
\be\label{eq:brinck}
\sup_{x>0}\int_x^{x+1} q_-(t)\, dt<\infty,\quad \sup_{x>0}\sum_{x_k\in[x,x+1]}\gA_k^-<\infty,
\ee
where $q_-=(|q|-q)/2$ and $\gA_k^-=(|\gA_k|-\gA_k)/2$. Then:
\begin{itemize}
\item[(i)] The form $\gt_{X,\gA,q}^0$ is lower semibounded.
\item[(ii)] If both the potential $q$ and the sequence $\gA$ are negative, then the condition \eqref{eq:brinck} is necessary and sufficient for the form $\gt_{X,\gA,q}^0$ to be lower semibounded.
\end{itemize}
\end{theorem}

\begin{remark}\label{rem:2.16}
(i) Theorem \ref{th:brinck} immediately implies that the operator $\rH_{X,\gA,q}$ is self-adjoint and lower semibounded if  $q$ is lower semibounded, $q(x)\ge -c$, and $\gA$ is a positive sequence.

(ii) The condition \eqref{eq:brinck} is only sufficient if $q$ and $\gA$ take values of both signs. Examples of $\gA$ and $q$, which do not satisfy \eqref{eq:brinck} but such that the operator is lower semibounded can be found in \cite{Bri59}, \cite[Example 2]{Bra85} (see also Example \ref{ex:2.15} below).

(iii) If conditions \eqref{eq:brinck} are satisfied, then  $\dom(\rH_{X,\gA,q}) $ is continuously embedded into $W^{1,2}(\R_+)$ (cf. Remark \ref{rem:2.1}),
%
%
\begin{equation}\label{2.24}
\dom(\rH_{X,\gA,q}) \hookrightarrow  \dom(\gt_{X,\gA,q}) \hookrightarrow  W^{1,2}(\R_+).
\end{equation}
Apparently  this embedding might be false even in the case of semibounded (hence self-adjoint) Hamiltonians  $\rH_{X,\gA,q}$. For further results and examples in the case of Hamiltonian $\rH_{q}$ with locally integrable potentials we refer to \cite{Eve83}, \cite{EveGieWei73}, \cite{Kal74}.
\end{remark}

\begin{example}[\cite{Bra85}]\label{ex:2.15}
Choose any $a>1$ and set $x_{2k-1}=k$ and $x_{2k}=k+a^{-3k}$. Let also $\gA_{2k-1}=a^k$ and $\gA_{2k}=-a^k$. Clearly,
\[
\sum_{x_k\in [n,n+1]}\gA_k^-=a^n \to +\infty\quad \text{as}\quad n\to \infty,
\]
and hence the second condition in \eqref{eq:brinck} is not fulfilled. However (see \cite[Example 2]{Bra85}), the Hamiltonian
\be
\rH:= \rH_{X,\gA,} =-\frac{d^2}{dx^2}+\sum_{k=1}^\infty \gA_k\delta(x-x_k)
\ee
is lower semibounded and hence  self-adjoint in $L^2(\R_+)$.

Let us also mention that the function $f(x)=a^{-x/2}$ is in the form domain, that is, $\gt_{\rH}[f]<\infty$. Moreover, $f\in W^{1,2}(\R_+)$. However,
\[
\sum_{k\in \N}\gA_k^-|f(x_k)|^2=\sum_{k=1}^\infty a^k a^{-k-a^{-3k}}=\infty.
\]
Note that in this example the embedding \eqref{2.24} holds true and is continuous \cite{Bra85}.
     \end{example}

%

Let us note that in the special case when there is a positive lower bound between interactions centers, i.e., $\gd_*>0$, the criterion obtained in Theorem \ref{th:km_sa}(ii) immediately implies the following statement.

\begin{corollary}\label{cor:2.13}
 Let  $q$ be bounded and $\gd_*>0$. Then the operator $\rH_{X,\gA,q}$ is lower semibounded precisely if so is the sequence $\gA$.
\end{corollary}

\begin{remark}
Note that Corollary \ref{cor:2.13} was first established by Brasche \cite{Bra85} by using the form approach. A different proof based on the boundary triplets approach was given in \cite{Mih94} (see also \cite{KosMal09}).
\end{remark}

\subsection{Spectral types}\label{sec:2.4} The literature on characterization of spectral types of Hamiltonians with $\delta$-interactions is enormous and for a comprehensive treatment of operators with periodic potentials, short range perturbations etc. we refer the reader to the monograph \cite{AGHH05}. In this subsection we shall review recent developments in the case $\gd_*=0$.

{\bf Discreteness.} We begin with the criteria for the operator $\rH_{X,\gA,q}$ to have a discrete spectrum. First of all, let us mention that the analog of the classical A.M.~Molchanov  discreteness criterion \cite{Mol53} (see also \cite{Bri59}, \cite{Gla65}) holds true.

\begin{theorem}[\cite{AlbKosMal09}]\label{th:molchanov}
Assume that the negative parts of $q$ and $\gA$ satisfy conditions \eqref{eq:brinck}. Then the lower semibounded operator $\rH_{X,\gA,q}$ has purely discrete spectrum if and only if for every $\varepsilon>0$
\be
\int_x^{x+\varepsilon} q(t)\, dt+\sum_{x_k\in[x,x+\varepsilon]}\gA_k\to +\infty\quad \text{as}\quad x\to+\infty.
\ee
\end{theorem}

In the case $q\in L^\infty$, we immediately arrive at the following result.

\begin{corollary}
If $q\in L^\infty(\R_+)$ and $\gA$ satisfies the second condition in \eqref{eq:brinck}, then the operator $\rH_{X,\gA,q}$ has purely discrete spectrum if and only if for every $\varepsilon>0$
\be
\sum_{x_k\in[x,x+\varepsilon]}\gA_k\to +\infty\quad \text{as}\quad x\to+\infty.
\ee
In particular, the spectrum of $\rH_{X,\gA,q}$ is purely discrete whenever
\be\label{eq:2.24D}
\gd_k\to 0 \quad \text{and}\quad \frac{\gA_k}{\gd_k}\to +\infty\quad \text{as}\quad k\to\infty.
\ee
\end{corollary}

Note that Theorem \ref{th:molchanov} applies only in the case of lower semibounded operators. Thus Theorem \ref{th:km_sa}(iii) completes Theorem \ref{th:molchanov} in the non lower semibounded case. In particular, applying the Chihara condition to the matrix
$B_{X,\gA}$, we arrive at the following result.

\begin{corollary}[\cite{KosMal09}]\label{th:km_discrA}
Let $q\in L^\infty(\R_+)$ and $X$ be such that $\gd^*<\infty$. Let the matrix $B_{X,\gA}$ be self-adjoint and let also $\gd_k\to 0$ and
\be\label{eq:2.23}
\lim_{k\to \infty}\frac{|\gA_k|}{\gd_k}=\infty\quad \text{and}\quad \lim_{k\to\infty}\frac{1}{\gA_k\gd_k}>-\frac{1}{4}.
\ee
Then the self-adjoint operator $\rH_{X,\gA,q}$ has a purely discrete spectrum.
\end{corollary}

Using the approach developed in \cite{Ism85} for smooth potentials, Ismagilov and Kostyuchenko \cite{IsmKos10} obtained the following result.

\begin{proposition}[\cite{IsmKos10}]\label{prop:ik_disc}
Assume that $q\in L^\infty(\R_+)$, $\gA_k<0$ for all $k\in \N$, $\gd_k\to0$ and
\be\label{eq:2.23b}
\frac{|\gA_k|}{\gd_k+\gd_{k+1}}-\frac{2}{\gd_k\gd_{k+1}}\to +\infty,\quad k\to\infty.
\ee
Then the operator $\rH_{X,\gA,q}$ is non lower semibounded and its spectrum is purely discrete.
\end{proposition}

Simple examples show that the condition \eqref{eq:2.23} does not imply  \eqref{eq:2.23b} and visa versa. Thus Corollary \ref{th:km_discrA} and Proposition \ref{prop:ik_disc} complete each other.

Let us mention that Corollary \ref{th:km_discrA} and Proposition \ref{prop:ik_disc} enable us to construct examples of Hamiltonians $\rH_{X,\gA,q}$, which are non lower semibounded, self-adjoint and their spectra are purely discrete.

\begin{example}[\cite{KosMal09b}]
Set $x_k=2\sqrt{k}$ and $\gA_k=-C\sqrt{k}$ with $C\in \R_+\setminus \{4\}$. Then $\gd_k=x_k-x_{k-1}\sim \frac{1}{\sqrt{k}}$, $k\to \infty$. The spectrum of the Hamiltonian
\[
\rH=-\frac{d^2}{dx^2}-C\sum_{k\in\N}\sqrt{k}\delta(x-2\sqrt{x})
\]
is non lower semibounded. Moreover, the spectrum is discrete if and only if $C>4$. The latter, in particular, implies that the second condition in \eqref{eq:2.23} is sharp.

\end{example}

\begin{remark}\label{rem:2.22}
 Combining Theorem \ref{th:molchanov} and Theorem \ref{th:km_sa}(iii) one can obtain a discreteness criterion for Jacobi matrices (for further details see \cite[\S 7]{AlbKosMal09}). This topic has attracted some attention recently, see \cite{CojJan07}, \cite{JanNab01}, \cite{JanNab03} and references therein.
\end{remark}

{\bf Continuous spectrum.} The next result is the extension of Birman's stability result \cite{Bir61} to the case of $\delta$-potentials.

\begin{theorem}[\cite{AlbKosMal09}]\label{th:birman}
Assume that the negative part of the potential $q$ satisfies \eqref{eq:brinck}. Then $\sigma_{\ess}(\rH_{X,\gA,q})=\sigma_{\ess}(\rH_q)$ provided that
\be
\lim_{x\to\infty}\sum_{x_k\in [x,x+1]}|\gA_k|\to 0.
\ee
 In particular, if in addition $q\to 0$ as $x\to \infty$, then $\sigma_{\ess}(\rH_{X,\gA,q})=[0,+\infty)$.
\end{theorem}

\begin{corollary}
If the negative part of the potential $q$ satisfies \eqref{eq:brinck} and
\be
\lim_{k\to\infty}\frac{\gA_k}{\gd_k}=0,
\ee
then $\sigma_{\ess}(\rH_{X,\gA,q})=\sigma_{\ess}(\rH_q)$.
\end{corollary}

\begin{remark}
In the case $\gd_*>0$, the condition $\gA_k\to 0$ as $k\to \infty$ is sufficient for the equality $\sigma_{\ess}(\rH_{X,\gA,q})=\sigma_{\ess}(\rH_q)$ to hold. However, if $\gd_*=0$, then this conclusion is no longer true. It might even happen that $\sigma(\rH_{X,\gA})$ is purely discrete. For example, it suffices to set $x_k=\sqrt{k}$ and $\gA_k=\frac{1}{k^{\varepsilon}}$ with $\varepsilon\in (0,\frac{1}{2})$, $k\in\N$ (cf. condition \eqref{eq:2.24D} and also \cite[Example 5.19]{KosMal09}).
\end{remark}

{\bf Absolutely continuous and singular spectra.} Theorem \ref{th:birman} can be specified under additional assumptions on $\gA$ and $X$.

\begin{theorem}[\cite{ShuSto94, Mih94, KosMal09}]\label{th:ac}
Assume that $\gd^*<\infty$ and $q\in L^\infty(\R_+)$.

\begin{itemize}
\item[(i)] Then $\sigma_{\ac}(\rH_{X,\gA,q})=\sigma_{\ac}(\rH_{q})$ provided that
  \begin{equation}\label{2.28A}
\sum_{k=1}^\infty\frac{|\gA_k|}{\gd_{k+1}}<\infty.
  \end{equation}
If in addition  $q\in L^1(\R_+)$, then
$\sigma_{\ac}(\rH_{X,\gA,q}) = [0,+\infty)$.

\item[(ii)] If  $q\equiv 0$,  $\gd_*>0$ and \eqref{2.28A} is satisfied, then $\sigma(\rH_{X,\gA})$
is purely absolutely continuous in $ (0,+\infty)$.
\end{itemize}
      \end{theorem}
\begin{remark}
The first statement of Theorem \ref{th:ac} is immediate by combining \cite[Corollary 5.15]{KosMal09} with the Kato--Rozenblum theorem \cite{Kato66}. Under an additional assumption $\gd_*>0$ this statement was proved in \cite{Mih94}. The second part of Theorem \ref{th:ac} was established in \cite{ShuSto94}.
\end{remark}

Let us also present one result on the absence of absolutely continuous spectrum.

\begin{theorem}[\cite{ShuSto94, Mih96}]\label{th:ac=0}
Let $X$ be such that $\gd_*>0$. Then $\sigma_{\ac}(\rH_{X,\gA,q})=\emptyset$ if at least one of the following conditions is satisfied:
\begin{itemize}
\item[(i)] $q$ is bounded from below, $\gA_k\ge 0$ for all $k\in\N$ and
\[
\limsup_{k\to \infty}\gA_k=+\infty,
\]
\item[(ii)]  $q\in L^\infty(\R_+)$ and
\[
\limsup_{k\to \infty}|\gA_k|=+\infty.
\]
\end{itemize}
\end{theorem}

The first and the second parts of Theorem \ref{th:ac=0} were established in \cite{ShuSto94} and \cite{Mih96}, respectively, by using a trace class technique similar to \cite{SimSpe89}.
   \begin{remark}[\cite{Mih96}]
Let $\omega$ be a Gaussian measure on the set of all real sequences. Then the subset of sequences
which are semibounded (below or above) has a zero measure (see \cite[\S 3.5]{ShiTyn67}). Therefore,  Theorem \ref{th:ac=0}
implies that for any fixed $X$ with $\gd_*>0$ the set of Hamiltonians $\rH_{X,\gA,q}$ having nonempty
absolutely continuous spectrum is of measure zero.
   \end{remark}

The following result was obtained by Lotoreichik \cite{Lot11}.

\begin{theorem}[\cite{Lot11}]\label{th:lot}
Assume that the set $X$ is sparse, that is 
\be
\lim_{k\to\infty}\frac{\gd_k}{\gd_{k-1}}=\infty.
\ee
Assume also that the intensities $\{\gA_k\}_{1}^\infty$ are such that $\gA_k\to \infty$ and
\be
\liminf_{k\to \infty}\frac{\gd_k}{\gd_{k-1}\gA_k^2}=:a\in (0,\infty)\cup\{\infty\}.
\ee
If $a\in(0,\infty)$, then:
\begin{itemize}
\item[(i)] $\sigma_{\ac}(\rH_{X,\gA})=\emptyset$,
\item[(ii)] $\sigma_{\rm pp}(\rH_{X,\gA})\subseteq [0,a^{-1}]$,
\item[(iii)] $[a^{-1},\infty]\subseteq \sigma_{\rm sc}(\rH_{X,\gA})\subseteq[0,+\infty)$.
\end{itemize}
If $a=\infty$ and all $\gA_k>0$, then $\sigma(\rH_{X,\gA})=\sigma_{\rm sc}(\rH_{X,\gA})=[0, +\infty)$.
\end{theorem}

%
%

{\em Further notes:} In \cite{IsmKos10}, Ismagilov and Kostyuchenko constructed a class of operators $\rH_{X,\gA}$ with purely point spectra having precisely two accumulation points  $0$ and  $+\infty$. Note that Hamiltonians with $\delta$-interactions form a good source of examples with exotic spectral properties. For example, Pearson in \cite[\S 14.6]{Pear89} used Schr\"odinger operators with $\delta$-interactions for constructing Hamiltonians with purely singular continuous spectrum.
Let us also mention papers \cite{GorHolMol07} and \cite{GorMolTsa91} for further examples of Schr\"odinger operators  having exotic spectra.

\subsection{Negative spectrum} During the last decade the problem on the number of negative eigenvalues for  Schr\"odinger operators with $\delta$-interactions attracted some attention. It is easy to observe from \eqref{eq:2.12}--\eqref{eq:2.14} that $\kappa_-(\rH_{X,\gA})\le \kappa_-(\gA)$, where $\kappa_-(\gA)$ is the number of negative entries in the sequence $\gA$. However, the converse inequality is, in general, not true.

Albeverio and Nizhnik discovered  in \cite{AlbNiz03} the connection between this problem and certain continued fractions. The latter enabled them to construct the algorithm for computing the number of negative eigenvalues.  Assuming that $\gd_*>0$ and using the boundary triplets approach, their construction has been extended in \cite{GolOri10} to the case of infinitely many $\delta$-interactions (cf. Theorem \ref{th:km_sa}(iv)). Note that a different matrix is used in \cite{AlbKosMalNei12}  for the analysis of $\kappa_-(\rH_{X,\gA})$. For simplicity we restrict our considerations to the case of finitely many point interactions.


%
 \begin{proposition}[\cite{AlbKosMalNei12}]\label{prop:akmn}
 If $X=\{x_k\}_{k=1}^N$ and $\gA=\{\gA_k\}_{k=1}^N$, where $N\in\N$, then
\be\label{eq:2.30}
\kappa_-(\rH_{X,\gA})=\kappa_+(M_{X,\gA}) - \kappa_+(\gA),
\ee
where
\be
M_{X,\gA}=
\left(\begin{array}{ccccc} \frac{1}{\gA_1}+x_1 & x_1 & x_1 & \dots  & x_1\\
x_1 & \frac{1}{\gA_2}+x_2 & x_2 & \dots & x_2\\
x_1 & x_2 & \frac{1}{\gA_3}+x_3 & \dots &x_3\\
 \dots& \dots& \dots& \dots & \dots \\
x_1 & x_2 & x_3 & \dots & \frac{1}{\gA_N}+x_N\\
\end{array}\right).
\ee
\end{proposition}

Next let us present the following extension of the celebrated Bargmann estimate (see, e.g., \cite{ReeSim78}).

  \begin{theorem}[\cite{AlbKosMalNei12}]\label{th:bargmann}
Let $q$, $X$ and $\gA$ be such that the operator $\rH_{X,\gA,q}$ is self-adjoint. If either $\gA^-\neq 0$ or $q_-\neq \mathbf{0}$, then
\be\label{eq:bargmann}
\kappa_-(\rH_{X,\gA,q})<\int_{\R_+} |q_-(x)|\, dx +\sum_{k=1}^\infty |\gA_k^-|x_k.
\ee
 \end{theorem}

Finally, let us mention that combining Theorem \ref{th:km_sa}(iv) with Theorem \ref{th:bargmann}, we arrive at the following estimate for Jacobi matrices.

\begin{corollary}\label{cor:2.30}
Let $X$ and $\gA$ be such that $\gd^*<\infty$ and $\gA^-\neq 0$. Let also $B_{X,\gA}$ given by \eqref{eq:B_xa} be self-adjoint. Then
\be\label{eq:2.33}
\kappa_-(B_{X,\gA})=\kappa_-(\rH_{X,\gA})<\sum_{k=1}^\infty |\gA_k^-|x_k.
\ee
\end{corollary}

Several different proofs of Theorem \ref{th:bargmann} can be found in \cite{AlbKosMalNei12}.
Let us give a proof of Corollary \ref{cor:2.30} for the case of a finite number of $\delta$-interactions based on Proposition \ref{prop:akmn}.

\begin{proof}
Firstly, assume that all $\gA_k$ are negative, that is $\gA=\gA^-$. Then $\kappa_+(\gA)=0$ and hence, by \eqref{eq:2.30}, we get
\[
\kappa_-(\rH_{X,\gA})= \kappa_+(M_{X,\gA}).
\]
Set $\Lambda:=\diag(|\gA_1|,\dots,|\gA_N|)$. Then we obtain from \eqref{eq:2.30}
\[
\kappa_+(M_{X,\gA})=\kappa_-(I_N-\Lambda^{1/2}M_X\Lambda^{1/2}),\quad
M_X=\left(\begin{array}{cccc} x_1 &  x_1 & \dots  & x_1\\
x_1 & x_2  &\dots & x_2\\
 \dots& \dots& \dots&  \dots \\
x_1 & x_2  & \dots & x_N\\
\end{array}\right).
\]
Therefore, denoting $M_X^\Lambda:=\Lambda^{1/2}M_X\Lambda^{1/2}$, we conclude that
\begin{align*}
\kappa_-(\rH_{X,\gA})=\kappa_+(M_{X,\gA})\le \sum_{\lambda_j(M_X^\Lambda)>1}1
<
\sum_{\lambda_j(M_X^\Lambda)>1}\lambda_j(M_X^\Lambda)\le {\rm tr}\, M_X^\Lambda=\sum_{k=1}^N |\gA_k|x_k.
\end{align*}
To prove the statement in the  case $\gA\neq \gA^-$,  it suffices to note
that $\kappa_-(\rH_{X,\gA})\le \kappa_-(\rH_{X,\gA^-})$.  
\end{proof}
\begin{remark}
(i) Let us mention that Theorem \ref{th:km_sa}(iv) and Proposition \ref{prop:akmn} 
enables us to construct the operator $\rH_{X,\gA}$ having a given number of negative eigenvalues (for further details see \cite{Ogu08}, \cite{Ogu10}, \cite{GolOri10}, \cite{AlbKosMalNei12}).

(ii) The above results demonstrate that Bargmann's bound is a one-sided estimate if the number of $\delta$-interactions is greater than $1$ (see \cite[Example 4.10]{AlbKosMalNei12} and also examples below).
\end{remark}

\begin{example}
Let $N\ge 2$. Assume that $\gA_1<0$ and $\gA_k>0$ for all $k\ge 2$. Clearly, $\kappa_+(\gA)=N-1$ and $\kappa_+(M_{X,\gA})\ge N-1$. If $\frac{1}{\gA_1}+x_1< 0$, then we immediately conclude that $\kappa_+(M_{X,\gA})= N-1$ and hence, by \eqref{eq:2.30}, $\kappa_-(\rH_{X,\gA})=0$, i.e., the operator is positive. Note that in this case the positivity also follows from the Bargmann estimate \eqref{eq:2.33}.

Next, if $\frac{1}{\gA_1}+x_1> 0$ and
\be\label{eq:2.34}
\Delta_2:=\det\begin{pmatrix}
\frac{1}{\gA_1}+x_1 & x_1\\
x_1 & \frac{1}{\gA_2}+x_2
\end{pmatrix}< 0,
\ee
then again we conclude $\kappa_+(M_{X,\gA})= N-1$ and hence, by \eqref{eq:2.30}, $\kappa_-(\rH_{X,\gA})=0$. Notice that in this case the Bargmann estimate \eqref{eq:2.33} only gives the inequality $\kappa_-(\rH_{X,\gA})\le 1$. Let us also mention that under the additional assumption $N=2$, the positivity of the determinant in \eqref{eq:2.34} implies that $\kappa_-(\rH_{X,\gA})= 1$.
\end{example}

\begin{example}
Let $N\ge 3$. Assume that $\gA_1<0$, $\gA_2<0$ and $\gA_k>0$ for all $k\ge 3$. Clearly, $\kappa_+(\gA)=N-2$ and $\kappa_+(M_{X,\gA})\ge N-2$. If $\frac{1}{\gA_1}+x_1< 0$ and the determinant in \eqref{eq:2.34} is positive, then $\frac{1}{\gA_2}+x_2< 0$ and  $\kappa_+(M_{X,\gA})= N-2$. Therefore, \eqref{eq:2.30} yields the equality $\kappa_-(\rH_{X,\gA})=0$. On the other hand, the Bargmann estimate \eqref{eq:2.33} only provides the inequality $\kappa_-(\rH_{X,\gA})< 2$.
\end{example}

%

\section{Hamiltonians with $\delta'$-interactions}\label{Sec:III}

\subsection{Definition of $\delta'$-interactions}
The main object of this section is the Hamiltonian formally given by the differential expression \eqref{I_01B}.
The existence of the model \eqref{I_01B} was pointed out in 1980 by Grossmann, Hoegh--Krohn and Mebkhout \cite{GroHoeMeb80}.
However, the first rigorous treatment of \eqref{I_01B} was made by Gesztesy and Holden in \cite{GesHol87}. Namely, they defined the Hamiltonian $\rH_{X,\gB,q}$ by using the method of boundary conditions. To be precise, let us consider \eqref{I_01B} on the interval $[0,b)$, $0<b\le +\infty$, assuming that the sequence $X=\{x_k\}_1^\infty$ is strictly increasing and accumulates at $b$. Then define the operator
 \begin{align}
\rH_{X,\gB,q}^0f:=&\tau_qf=-f''+q(x)f,\quad f\in\dom(\rH_{X,\gB,q}^0),\label{eq:h_b1}\\
\dom(\rH_{X,\gB,q}^0):=&\Big\{f\in W^{2,1}_{\comp}([0,b)\setminus X):\, f(0)=0,\Big. \nonumber
\\ &\Big.  \begin{array}{c}
f'(x_k+)=f'(x_k-)\\ f(x_k+)-f(x_k-)=\gB_k f'(x_k) \end{array},\, \tau_q f\in L^2(\R_+)\Big\}.\label{eq:h_b2}
\end{align}
Clearly, $\rH_{X,\gB,q}^0$ is symmetric. Let us denote its closure by $\rH_{X,\gB,q}$:
\be\label{eq:h_b3}
\rH_{X,\gB,q}:=\overline{\rH_{X,\gB,q}^0}.
\ee
For $q=\bold{0}$ we set $\rH_{X,\gB}:=\rH_{X,\gB,\bold{0}}$. If $\gB_k= \infty$, then the boundary condition at $x_k$ reads as
$f'(x_k+)=f'(x_k-) = 0$. Therefore, the operator $\rH_{X,\infty,q}$ becomes
\be\label{eq:h_N}
 \rH_{X,\infty,q}:=\rH_{X,q}^N = \bigoplus_{k\in \N} \rH_{q,k}^N, \quad  \dom(\rH_{X,q}^N) = \bigoplus_{k\in \N}\dom(\rH_{q,k}^N),
\ee
where $\rH_{q,k}^N$ is the Neumann realization of $\tau_q=-\frac{d^2}{dx^2}+q$ in $L^2(x_{k-1},x_k)$.

Up to now it was not clear how to apply the form approach in order to rigorously define a $\delta'$-interaction on $X$ (cf. \cite[Section 7.2]{Exn08}). Indeed, a very naive guess is to consider  a single $\delta'$-interaction at $x_0$  as the following form sum
\[
\gt'[f]=\int_{\R} |f'(x)|^2\, dx +\gB_0|f'(x_0)|^2
\]
defined on the domain
\[
 \dom(\gt')=\{f\in W^{1,2}(\R):\, f'(x_0) \ \text{exists and is finite}\}.
\]
Clearly, the form $\gt'$ is not closable. However (see \cite{KosMal12}\footnote{In the paper \cite{BehLanLot12}, which appeared during the preparation of \cite{KosMal12}, Hamiltonians with a $\delta'$-interaction supported on a hypersurface  are treated in a similar way.}),
 one needs to consider a $\delta'$-interaction as a form sum of two forms $\gt_N$ and $\gb$, where
\be\label{eq:3.5}
\gt_N[f]:=\int_{\R} |f'(x)|^2\, dx,\quad \dom(\gt_N):=W^{1,2}(\R\setminus\{x_0\}),
\ee
and
\be\label{eq:3.6}
\gb[f]:=\frac{|f(x_0+)-f(x_0-)|^2}{\gB_0},\quad \dom(\gt_N):=W^{1,2}(\R\setminus\{x_0\}).
\ee
Let us note that the operator
\be
\rH_{x_0}^N:=-\frac{d^2}{dx^2}, \quad \dom(\rH_{x_0}^N)=\{f\in W^{2,2}(\R\setminus \{x_0\}):\ f'(x_0+)=f'(x_0-)=0\},
\ee
is associated with the form $\gt_N$. Clearly, $\rH_{x_0}^N$ is the direct sum of Neumann realizations of $-\frac{d^2}{dx^2}$ in $L^2(-\infty,x_0)$ and $L^2(x_0,+\infty)$, respectively. Note that the form $\gb$ is infinitesimally form bounded with respect to the form $\gt_N$ and hence, by the KLMN theorem,  the form
\be\label{eq:3.7}
\gt'[f]:=\gt_N[f]+\gb[f], \quad \dom(\gt'):=W^{1,2}(\R\setminus\{x_0\}),
\ee
is closed and lower semibounded and gives rise to a self-adjoint operator 
\be\label{eq:3.8}
\begin{split}
\rH'&=-\frac{d^2}{dx^2},\\ \dom(\rH')&:=\Big\{f\in W^{2,2}(\R\setminus\{x_0\}): \begin{array}{c}
f'(x_0+)=f'(x_0-)\\
f(x_0+)-f(x_0-)=\gB_0f'(x_0+)\end{array}\Big\}.
\end{split}
\ee
\begin{remark}
(i) Let us emphasize that the definition of a $\delta'$-interaction via the form sum \eqref{eq:3.7} allows to observe the key difference between $\delta$ and $\delta'$-interactions. Namely, $\delta$-interactions are considered as a perturbation of the free Hamiltonian. However, $\delta'$-interactions can be viewed as a perturbation of the operator $\rH_{X,q}^N$ defined by \eqref{eq:h_N}. In particular, in the case of infinitely many interaction centers, the free Hamiltonian has purely absolutely continuous spectrum though the spectrum of $\rH_{X,q}^N$ is purely point. Let us also mention that the idea to consider Hamiltonians with $\delta'$-interactions $\rH_{X,\gB,q}$ as a perturbation of the Neumann realization $\rH_{X,q}^N$ was used by Exner in \cite{Exn95} in order to prove that the spectra of $\delta'$ Wannier--Stark Hamiltonians have no absolutely continuous parts.

(ii) 
 Similar to $\delta$-interactions, Hamiltonians with $\delta'$-interactions can also be considered as quasi-differential operators. For example, set $p(x)=x+\gB_0\chi_{[x_0,+\infty)(x)}$ and consider in $L^2(\R)$ the following differential expression
$
\tau_p:=-\frac{d}{dx}\frac{d}{dp(x)}.
$
It can be shown (cf. \cite[\S 3]{EckTes12} and \cite{EckKosMalTes12}), that the corresponding self-adjoint operator coincides with $\rH'$ given by \eqref{eq:3.8}. Note that this definition enables us to introduce $\delta'$-interaction on an arbitrary set of Lebesgue
measure zero and this will be done in the forthcoming paper \cite{EckKosMalTes12}. Let us also mention that using a different approach these operators have been studied recently by Albeverio and Nizhnik \cite{AlbNiz06} and Brasche and
Nizhnik \cite{BraNiz11}.
\end{remark}

As in the case of $\delta$-interactions, the domain of $\rH_{X,\gB,q}$ can be further specified if $q\in L^\infty(\R_+)$. Let us equip $\dom(\rH_{X,\gB,q})$
with the graph norm of $\rH_{X,\gB,q}$.
   \begin{proposition}\label{prop:emb'}
Let $q\in L^\infty(\R_+)$. Then:

\begin{itemize}
\item[(i)]  The operator $\rH_{X,\gB,q}$ is self-adjoint and its domain is given by
 \be
 \dom\big(\rH_{X,\gB,q}\big):=\Big\{f\in  W^{2,2}(\R_+\setminus X):\,  f(0)=0, 
\begin{array}{c}
f'(x_k+) = f'(x_k-)\\ f(x_k+)-f(x_k-) = \gB_k f(x_k) \end{array}\Big\}.\label{eq:2.5BB}
 \ee
\item[(ii)] The embedding
\begin{align*}
 W^{2,2}(\R_+\setminus X) \hookrightarrow  W^{1,2}(\R_+\setminus X), 
\end{align*}
holds and is continuous  if and only if $\gd_*>0$.
\item[(iii)] If $\gd_*>0$, then
the embedding %
  \begin{align}\label{2.8B}
\dom(\rH_{X,\gB,q}) 
\hookrightarrow  W^{1,2}(\R_+\setminus X) 
    \end{align}
holds  and  is continuous.
\end{itemize}
\end{proposition}

\begin{remark}
Self-adjointness of $\rH_{X,\gB,q}$ was established in \cite{BusStoWei95} (see also Section \ref{ss:3.2}). The proof of Proposition \ref{prop:emb'} can be found in \cite{KosMal12}.
%
\end{remark}

{\em Further remarks:}  There is one more approach to define $\delta'$-interactions. Namely, a single $\delta'$-interaction can be treated as an $\cH_{-2}$-perturbation of the free Hamiltonian. For further details and results we refer to the monographs \cite{AlbKur00} and \cite{Kos93}.

Let us also mention that there is a difference between $\delta'$-interactions and $\delta'$-potentials. During the last few years there was some activity in understanding the Hamiltonians with $\delta'$-potentials. In this connection we refer to the recent papers \cite{GolMan09}, \cite{GolHry11}, \cite{GolHry11b} and \cite{BraNiz11} (see also the references therein).

\subsection{Self-adjointness and connection with Jacobi matrices}\label{ss:3.2} 
The first results on the self-adjointness for Hamiltonians with $\delta'$-interactions were obtained by Gesztesy and Holden \cite{GesHol87} (see also \cite[\S III.3]{AGHH05}). Using the approach introduced by Phariseau in \cite{Pha60} for $\delta$-interactions, Gesztesy and Holden \cite{GesHol87} established self-adjointness in the case  $q\equiv 0$ and $\gd_*>0$. Let us stress that the analysis becomes much more complicated if  either $\gd_*=0$ or $q\notin L^\infty$.

The next step was made by Buschmann, Stolz and Weidmann \cite{BusStoWei95}. Namely,
in contrast to Hamiltonians with $\delta$-interactions, it was observed in \cite{BusStoWei95} that the Hamiltonian $\rH_{X,\gB,q}$
is always self-adjoint provided that $q\in L^\infty$ and $b=+\infty$. However, as in the case of $\delta$-interactions, Buschmann, Stolz and Weidmann \cite{BusStoWei95} proved that $n_\pm(\rH_{X,\gB,q})\le 1$ and the deficiency indices can be characterized by using Weyl's limit point/limit circle criterion.

Using the boundary triplets approach, it was shown in \cite{KosMal09} that in the case $q\in L^\infty(0,b)$ certain spectral properties of $\rH_{X,\gB,q}$ are closely connected with those of the following Jacobi matrix
\be\label{eq:B_xa'}
B_{X,\gB}:=\left(\begin{array}{cccccc}
\gd_1^{-2} & \gd_1^{-2} & 0& 0& 0 & \dots\\
\gd_1^{-2} & \frac{\gd_1^{-1}}{\beta_1}+\gd_1^{-2} & \frac{\gd_1^{-1/2}\gd_2^{-1/2}}{\beta_1}& 0& 0 & \dots\\
0 & \frac{\gd_1^{-1/2}\gd_2^{-1/2}}{\beta_1} & \frac{\gd_2^{-1}}{\beta_1}+\gd_2^{-2}& \gd_2^{-2}& 0 & \dots\\
0 & 0 & \gd_2^{-2}& \frac{\gd_2^{-1}}{\beta_2}+\gd_2^{-2}& \frac{\gd_2^{-1/2}\gd_3^{-1/2}}{\beta_2} & \dots\\
0 & 0 & 0& \frac{\gd_2^{-1/2}\gd_3^{-1/2}}{\beta_2}& \frac{\gd_3^{-1}}{\beta_2}+\gd_3^{-2} & \dots\\
\dots & \dots & \dots& \dots& \dots & \dots
\end{array}\right).
\ee
More precisely, under a suitable choice of a boundary triplet $\Pi=\{l^2(\N),\Gamma_0,\Gamma_1\}$ for the operator $\rH_{X,q}^*$ (see Section \ref{ss:2.3}), the operator $\rH_{X,\gB,q}$ admits the following representation
\[
\dom(\rH_{X,\gB,q})=\{f\in \dom(\rH_{X,q}^*):\ \Gamma_1=B_{X,\gB}\Gamma_0\}.
\]
The next  result was established in \cite{KosMal09}, \cite{KosMal10}.

\begin{theorem}[\cite{KosMal09, KosMal10}]\label{th:km_sa'}
Let $\rH_{X,\gB,q}$ be given by \eqref{eq:h_b1}--\eqref{eq:h_b3} and let $B_{X,\gB}$ be the matrix \eqref{eq:B_xa'}. Let also $q\in L^\infty$ and $\gd^*<\infty$. Then:
\begin{itemize}
\item[(i)] $n_\pm(\rH_{X,\gB,q})=n_\pm(B_{X,\gB})$. In particular, $\rH_{X,\gB,q}$ is self-adjoint if and only if so is $B_{X,\gB}$.
\item[(ii)] The operator $\rH_{X,\gB,q}$ is lower semibounded if and only if so is $B_{X,\gB}$.
\item[(iii)] If $\rH_{X,\gB,q}$ is self-adjoint, then its spectrum is purely discrete if and only if $\gd_k\to 0$ and the spectrum of $B_{X,\gB}$ is purely discrete.
\item[(iv)] If $\rH_{X,\gB}$ is self-adjoint, then
\be
\kappa_-(\rH_{X,\gB})=\kappa_-(B_{X,\gB}).
\ee
\item[(v)] If $\rH_{X,\gB}$ is self-adjoint, then for any $p\in [1,\infty]$
\be
\rH_{X,\gB}^-\in \gS_p(L^2) \Longleftrightarrow B_{X,\gB}^-\in \gS_p(l^2).
\ee
In particular, $\sigma_{\ess}(\rH_{X,\gA})\subseteq\R_+$ if and only if $\sigma_{\ess}(B_{X,\gA})\subseteq\R_+$.
\end{itemize}
\end{theorem}

It is interesting to note that the matrix $B_{X,\gB}$ admits the  representation
\begin{equation}\label{eq:fac1}
B_{X,\gB}=R_X^{-1}(I+U)D_{X,\gB}^{-1}(I+U^*)R_X^{-1},
\end{equation}
where $U$ is the unilateral shift on $l^2(\N)$ and
\be\label{eq:fac2}
R_X=\bigoplus_{k=1}^\infty \begin{pmatrix}
\sqrt{\gd_k} & 0\\
0 & \sqrt{\gd_k}
\end{pmatrix},\quad D_{X,\gB}=\bigoplus_{k=1}^\infty \begin{pmatrix}
\gd_k & 0\\
0 & \gB_k
\end{pmatrix}.
\ee
 This observation immediately implies a connection of the Hamiltonian $\rH_{X,\gB,q}$ with Krein--Stieltjes string operators \cite[Appendix]{Akh}, \cite{KK71} (see also \cite[\S 6]{KosMal09} and \cite{KosMal10} for further details). Namely, if all $\gB_k$ are positive, then setting $l_{2k - 1} := \gd_k$, $l_{2k} := \gB_k$, $m_{2k -1} = m_{2k} := \gd_k$, $k\in\N$, the difference equation associated with the matrix $B_{X,\gB}$ describes the motion of an
inhomogeneous string (Krein--Stieltjes string) with the
mass distribution $\mathcal{M}(y) = \sum_{y_k<y} m_k$, where $y_k-y_{k-1}=l_k$ and $y_0=0$.
This class of matrices is studied sufficiently
well. In particular, applying Hamburger's Theorem \cite[Theorem 0.5]{Akh} to the matrix $B_{X,\gB}$,
we arrive at the following self-adjointness criterion.

 \begin{theorem}[\cite{KosMal09}]\label{prop_01}
Deficiency indices of the operator $\rH_{X,\gB}$ are equal and are not greater than one, $\mathrm{n}_+(\rH_{X,\gB})=\mathrm{n}_-(\rH_{X,\gB})\leq1$.
Furthermore, $\rH_{X,\gB}$ is self-adjoint if and only if at least one of the following conditions hold:
\begin{itemize}
\item[(i)] \quad $\sum_{n=1}^{\infty} \gd_n =\infty$, i.e., $b=+\infty$;
\item[(ii)] \quad $\sum_{n=1}^\infty \left[\gd_{n+1} \big|\sum_{i=1}^n (\beta_i+\gd_i)\big|^2\right]=\infty$.
     \end{itemize}
     \end{theorem}

     \begin{remark}
     As distinguished from the case of $\delta$-interactions, by Theorem \ref{prop_01}(i), the operator $\rH_{X,\gB}$ is self-adjoint in $L^2(\R_+)$ for any $\gB\subset \R$ (cf. Example \ref{ex:2.8}).  This fact was first observed in \cite{BusStoWei95}. Let us also mention that statement (ii) provides the self-adjointness criterion in the case of a finite interval $[0,b)$, $b<\infty$.
     \end{remark}

\subsection{Semiboundedness}

Let $b=+\infty$ and $X=\{x_k\}_{k=1}^\infty$ be a strictly increasing sequence accumulating at $+\infty$. Consider  the following energy form
\begin{align}
\gt_{X,\gB,q}^0[f]:=\int_{\R_+}(|f'(x)|^2+q(x)|f(x)|^2)\, dx +\sum_{k=1}^\infty \frac{|f(x_k+)-f(x_k-)|^2}{\gB_k},\label{eq:3.12}\\
\quad \dom(\gt^0_{X,\gB,q})=\{f\in W^{1,2}(\R_+\setminus X)\cap L^2_{\comp}(\R_+):\, \gt_{X,\gB,q}^0[f]<\infty\}.\label{eq:3.13}
\end{align}
Integrating by parts, one gets that the form $\gt_{X,\gB,q}^0$
admits the representation
\be\label{eq:3.14}
\gt_{X,\gB,q}^0[f]=(\rH_{X,\gB,q}^0f,f)_{L^2},\quad f\in 
\dom(\rH_{X,\gB,q}^0),
\ee
and hence the form $\gt_{X,\gB,q}^0$ is closable whenever it is lower semibounded.

Firstly, let us mention that similar to the case of $\delta$-interactions, the analog of the Glazman--Povzner--Wienholtz Theorem holds true in the case of $\delta'$-interactions.

\begin{theorem}[\cite{KosMal12}]\label{th:povzner'}
If the Hamiltonian $\rH_{X,\gB,q}$ is lower semibounded, then it is self-adjoint.
\end{theorem}

Combining this theorem with the representation \eqref{eq:3.14}, we immediately arrive at the following result.

\begin{corollary}
 If the form $\gt_{X,\gB,q}^0$ is lower semibounded, then it is closable and the self-adjoint operator associated with its closure $\gt_{X,\gB,q}:=\overline{\gt_{X,\gB,q}^0}$ coincides with $\rH_{X,\gB,q}$.
\end{corollary}

Next we state the counterpart of Theorem \ref{th:brinck}.

\begin{theorem}[\cite{KosMal12}]\label{th:km_sb'}
Assume that $d^*<\infty$ and  there exist positive constants $C_0$, $C_1>0$ such that
%
%
      \begin{equation}\label{eq:b_klmn}
\frac{1}{\gd_k}\int_{x_{k-1}}^{x_k}q_-(x)dx\le C_0, \quad
\frac{1}{\gB_k^-}\le C_1 \min\{\gd_{k},\gd_{k+1}\},\quad k\in\N.\footnote{Here we formally set $\frac{1}{\gB_k^-}:=0$ if $\gB_k^-=0$, i.e., the corresponding inequality automatically holds true if $\gB_k$ is positive.}
   \end{equation}
 Then: 
\begin{itemize}
\item[(i)] The form $\gt_{X,\gB,q}^0$ is lower semibounded and the Hamiltonian $\rH_{X,\gB,q}$ is lower semibounded and self-adjoint,

\item[(ii)] If both the potential $q$ and the sequence $\gB$ are negative, then the conditions \eqref{eq:b_klmn} are also necessary for the form $\gt_{X,\gB,q}^0$ (and hence for the operator $\rH_{X,\gB,q}$) to be lower semibounded.
\end{itemize}
\end{theorem}

\begin{remark}
(i) Theorem \ref{th:km_sb'} immediately implies that the operator $\rH_{X,\gB,q}$ is self-adjoint and lower semibounded if $q$ is bounded from below,  $q(x)\ge -c$, and $\gB$ is a positive sequence.

(ii) Note also that conditions \eqref{eq:b_klmn} are only sufficient if $q$ and $\gB$ take values of both signs. Let us also mention that \eqref{eq:b_klmn} imply the corresponding conditions \eqref{eq:brinck} for $q$. However, the converse is not true.
\end{remark}

Finally, let us present some simple conditions, which are necessary for the operator $\rH_{X,\gB,q}$ to be lower semibounded (for further conditions see \cite{KosMal09} and \cite{KosMal12}).

\begin{lemma}[\cite{KosMal12}]\label{lem:nec}
Let $q= \bold{0}$.
If the form  $\gt_{X,\gB}^0$ is lower semibounded, that is $\gt_{X,\gB}^0\ge -C$ for some $C\ge 0$, then:
\begin{enumerate}
\item[(i)] for all $\gB_k^-\neq 0$
  \begin{equation}\label{eq:inf_b}
\frac{1}{\gB_k^-}\le 1+\frac{C}{3},\qquad k\in\N,
  \end{equation}
\item[(ii)]
  \begin{equation}\label{eq:nec_1}
\frac{1}{\gB_j^-}:=\frac{1}{|\gB_{k_j}|}\le C\min \{\gd_j^-,\gd_{j+1}^-\},\qquad 
 \quad j\in\N,
  \end{equation}
where $X^-=\{x_j^-\}_{j=1}^\infty:=\{x_{k_j}\}$  is the subsequence supporting negative intensities and $\gd_j^-:=x_j^--x_{j-1}^-=x_{k_j}-x_{k_{j-1}}$.
\end{enumerate}
\end{lemma}

\subsection{Spectral types}
Hamiltonians with periodically arranged $\delta'$-interactions were first discussed by Gesztesy and Holden in \cite{GesHol87}. Namely, they investigated in great detail the spectral properties of $\rH_{X,\gB}$ in the cases when $\gA_k\equiv \gA\in \R$ and $X=a\Z$ (crystal) or $X=a\N$ (half-crystal). Also, in \cite{GesHol87}, it was studied how the introduction of impurities affects spectral properties of crystals. The analysis of various types of ordered alloys, both deterministic and random, for this model was done in \cite{GesHolKir87} where, e.g., the Saxon--Hunter conjecture \cite{SaxHun49},
 concerning gaps in the spectrum was proved.   For a comprehensive treatment of these models we refer to the monograph \cite{AGHH05}.
 The main aim of this subsection is to review recent developments in the case $\gd_*=0$.

{\bf Discreteness.} 
%
%
 Using Theorem \ref{th:km_sa'}(iii) and the Kac--Krein discreteness criterion \cite{KK58}, one can prove the following result.

 \begin{proposition}[\cite{KosMal09}]\label{prop:discr1}
Let $\cI=\R_+$ and $\gd_k\to 0$. The spectrum of the operator $\rH_{X,\gB}$ is not discrete if at least one of the following conditions hold:
\begin{itemize}
\item[(i)] \quad $\lim_{k\to\infty}x_k\sum_{j=k}^\infty\gd_j^3>0;$
\item[(ii)] \quad $\gB_k\geq -C\gd_k^3$, \quad $k\in\N$, \quad $C> 0;$
\item[(iii)] \quad $\gB_k^{-}\le -C(\gd_k^{-1}+\gd_{k+1}^{-1})$, \quad $k\in\N$, \quad $C> 0$. 
    \end{itemize}
    \end{proposition}
It follows from Proposition \ref{prop:discr1} that discreteness of the spectrum is a very rare property. For instance,  the spectrum of the operator $\rH_{X,\gB}$ is not discrete if either $\gB_n>0$ for all $n\in\N$ or $\{\gd_n\}_{n=1}^\infty\notin l^3(\N)$. However, it is possible to indicate certain conditions on $X$ and $\gB$ which guarantee the discreteness.

\begin{proposition}[\cite{KosMal09}]\label{prop:discr2}
Assume that $\gB_k+\gd_k\ge0$ for all $k\in\N$ and $\cI=\R_+$. Then the spectrum of $\rH_{X,\gB}$ is purely discrete if and only if
\begin{equation}\label{IV.3.4_09}
\lim_{k\to\infty}x_k\sum_{j=k}^\infty\gd_j^3=0\quad
\text{and}\quad
\lim_{k\to\infty}x_k\sum_{j=k}^\infty(\gB_j+\gd_j)=0.
\end{equation}
\end{proposition}


As it was already mentioned, the  Hamiltonian $\rH_{X,\gB,q}$ can be considered as a form sum perturbation of the operator 
\be\label{eq:HN}
\rH_{X,q}^N:=\bigoplus_{k\in\N} \rH_{q,k}^N,
\ee
where $\rH_{q,k}^N$ is the Neumann realization of $-\frac{d^2}{dx^2}+q(x)$ in $L^2(x_{k-1},x_k)$. The next result provides a discreteness criterion for the operator $\rH_{X,q}^N$.

\begin{theorem}[\cite{KosMal12}]\label{th:disc_N}
Assume that  $d^*<\infty$,  
$q\in L^1_{\loc}(\R_+)$, and $q_-$  satisfies the first condition in \eqref{eq:b_klmn}.  
Then the spectrum of $\rH_{X,q}^N$ is discrete if and only if the following conditions are
satisfied:
    \begin{equation}\label{Intro2.1}
\text{ for every $\varepsilon>0$}\quad \int^{x + \varepsilon}_x q(t)dt\   
\to +\infty \qquad \text{as}\qquad x\to \infty.
    \end{equation}
%
%
       \begin{equation}\label{Intro2.2_N}
\frac{1}{\gd_k} \int_{x_{k-1}}^{x_k} q(x) dx  \to +\infty \qquad \text{as}\qquad k\to \infty
    \end{equation}
           \end{theorem}
 
It is an immediate corollary of Theorem \ref{th:disc_N} that both conditions \eqref{Intro2.1} and \eqref{Intro2.2_N} are sufficient for the discreteness of the spectrum of $\rH_{X,\gB,q}$. Moreover, \eqref{Intro2.1} remains to be necessary although \eqref{Intro2.2_N} is no longer necessary. 

 \begin{theorem}[\cite{KosMal12}]\label{cor:3.78}
 Assume that  
$q\in L^1_{\loc}(\R_+)$,   $d^*<\infty$  
 and conditions  \eqref{eq:b_klmn} are  satisfied.  
 \begin{itemize}
 \item[(i)]  If $q$ satisfies \eqref{Intro2.1} and \eqref{Intro2.2_N}, then the spectrum of $\rH_{X,\gB,q}$ is discrete.
 \item[(ii)] If the spectrum of the lower semibounded Hamiltonian $\rH_{X,\gB,q}$ is purely discrete, then $q$ satisfies \eqref{Intro2.1} and 
 \be\label{eq:4.5'}
\frac{1}{\gd_k}\Big(\int_{x_{k-1}}^{x_k}q(x)dx+\frac{1}{\gB_{k-1}}+\frac{1}{\gB_k}\Big)\to +\infty.
\ee
 \end{itemize}
 \end{theorem}
         The next result complements 
         Proposition \ref{prop:discr1}. 
           \begin{proposition}[\cite{KosMal12}]\label{prop:discr3}
Let $b=+\infty$ and $q\in L^\infty(\R_+)$. If the Hamiltonian $\rH_{X,\gB,q}$ is lower semibounded, then its spectrum is not discrete. In particular, if the operator $\rH_{X,\gB}:=\rH_{X,\gB,\bold{0}}$ is lower semibounded, then its spectrum is not discrete.
\end{proposition}

There is a gap between necessary and sufficient conditions in Theorem \ref{cor:3.78}. Indeed, the next result shows that condition \eqref{Intro2.2_N} is only sufficient and in cases when \eqref{Intro2.2_N} is not satisfied the discreteness of $\sigma(\rH_{X,\gB,q})$ depends on $q$ and $\gB$. 
In particular, the spectrum of the Hamiltonian $\rH_{X,\beta, q}$ might be discrete although the spectrum of the corresponding Neumann realization $\rH_{X,q}^N$ is not.
\begin{proposition}[\cite{KosMal12}]\label{prop:new}
Let $X=\{x_k\}^{\infty}_1\subset\R_+$ be such that $\gd^*<\infty$ and
%
%
   \begin{equation}\label{new:01}
\inf_{k\in\N} \gd_{2k-1} =: \varepsilon_0>0 \qquad \text{and}\qquad \lim_{k\to\infty}d_{2k}=0.
\end{equation}
Let $q$ satisfy  \eqref{eq:b_klmn} and  Molchanov's condition \eqref{Intro2.1}. If $\gB$ satisfies \eqref{eq:b_klmn} and 
     \begin{equation}\label{new:02}
\lim_{k\to\infty}d_{2k}\beta_{2k-1}=0,
     \end{equation}
then the spectrum $\sigma(\rH_{X,\beta, q})$ of the Hamiltonian $\rH_{X,\beta, q}$ is purely discrete.
  \end{proposition}

{\bf Continuous spectrum.}
\begin{theorem}[\cite{KosMal12}]\label{thContSpec}
Assume that $q\in L^1_{\loc}(\R_+)$ and the first
condition in \eqref{eq:b_klmn}  is satisfied. Then $\sigma_{\ess}(\rH_{X,\beta,q})=\sigma_{\ess}(\rH_{X,q}^N)$ if
     \begin{equation}\label{eq:3.21B}
 \frac{|\beta_k|^{-1}}{\min\{d_k, d_{k+1}\}} \to 0 \quad
\text{as}\quad k\to \infty.
       \end{equation}

If, in addition,
     \begin{equation}\label{eq:3.21}
\lim_{k\to \infty}\frac{1}{d_k}\int_{x_{k-1}}^{x_k}|q(x)|dx = 0,
     \end{equation}
then
      \begin{equation}\label{1.12}
\sigma_{\ess}(\rH_{X,\beta,q})=\sigma_{\ess}(\rH_{X,q}^N) =
\sigma_{\ess}(\rH_{X}^N).
     \end{equation}
     \end{theorem}

Noting that the spectrum of  $\rH_{X,q}^N$ is pure point, we can construct various examples of operators $\rH_{X,\gB,q}$ with exotic essential spectra. In particular, \eqref{1.12} implies that the structure of $\sigma_{\ess}(\rH_{X,\gB,q})$ depends only on a "geometry"\ of $X$.

\begin{corollary}[\cite{KosMal12}]
Let the assumptions of Theorem \ref{thContSpec} be satisfied. Assume also that $q$ satisfies \eqref{eq:3.21} and $\lim_{k\to\infty}d_k = 0$.  Then
     \begin{equation}\label{1.13}
\sigma_{\ess}(\rH_{X,\beta,q})
=\{0\},
   \end{equation}
i.e. the spectrum of   $\rH_{X,\beta,q}$  is pure  point and accumulates only at $0$ and $\infty$.
\end{corollary}

\subsection{Negative spectrum} In contrast to the case of $\delta$-interactions, the number of negative squares is determined by the number of negative intensities.

\begin{theorem}[\cite{GolOri10, KosMal10}]\label{th:kappa-}
If the operator $\rH_{X,\gB}$ is self-adjoint, then
\be\label{eq:kappa-}
\kappa_-(\rH_{X,\gB})=\kappa_-(\gB).
\ee
In particular, the operator $\rH_{X,\gB}$ is nonnegative if and only if  $\gB_k\ge 0$, $k\in \N$.
\end{theorem}

\begin{proof}
By Theorem \ref{th:km_sa'}(iv), we get $\kappa_-(\rH_{X,\gB}) = \kappa_-(B_{X,\gB})$. On the other hand, it follows from  the factorization \eqref{eq:fac1}--\eqref{eq:fac2} that $\kappa_-(B_{X,\gB})= \kappa_-(\gB)$. Combining both equalities we complete the proof.
\end{proof}

\begin{remark}
The equality \eqref{eq:kappa-} was observed in \cite{AlbNiz03b} in the special case when $|X|=N<\infty$ and all intensities are negative.  In the case $\gd_*>0$, Theorem \ref{th:kappa-}  was established in \cite{GolOri10}. The assumption $\gd_*>0$ was removed in \cite{KosMal10} by using a different method.
\end{remark}

Finally, let us mention that Theorem \ref{th:kappa-} enables us to give a different proof of Corollary \ref{prop:discr3}.

\begin{proof}[Proof of Corollary \ref{prop:discr3}]
Clearly, it suffices to prove Corollary \ref{prop:discr3} in the case $q\equiv 0$. By Proposition \ref{prop:discr1}(ii), if  the spectrum of $\rH_{X,\gB}$ is purely discrete, then $\kappa_-(\gB)=\infty$. Therefore, by Theorem \ref{th:kappa-}, $\kappa_-(\rH_{X,\gB})=\infty$. However, if $\rH_{X,\gB}$ is lower semibounded, then the negative spectrum of $\rH_{X,\gB}$ has at least one finite accumulation point. This contradiction completes the proof.
\end{proof}

{\em Further notes:} In \cite{Niz03}, \cite{AlbNiz06}, \cite{BraNiz11}, Nizhnik with co-authors introduced $\delta'$-interactions on sets of a Lebesgue measure zero, for example, on Cantor type sets. In these papers, the self-adjointness and basic spectral properties of these operators have been analyzed.  A different approach to analyze the spectral properties of Hamiltonians with $\delta'$-interactions on Cantor type sets is proposed in \cite{EckKosMalTes12}.


%
%
%
%
%
%


{\bf Acknowledgments.}
The authors are grateful to Rostyslav Hryniv for the careful reading of the manuscript and helpful hints with respect to the literature.
We are also grateful to Gerald Teschl and the anonymous referee for useful remarks.



\end{document}